\documentclass{article}

\usepackage{amsthm,amsmath,amssymb}
\usepackage[numbers]{natbib}

\usepackage{hyperref}
\hypersetup{
    colorlinks,%
    citecolor=black,%
    filecolor=black,%
    linkcolor=black,%
    urlcolor=black
}
\usepackage{hypernat}
\usepackage[top=2.5cm, bottom=2.5cm, left=2.8cm,
right=2.8cm]{geometry} 
\usepackage{verbatim}
\usepackage{graphicx}
\usepackage{caption}
\usepackage{subcaption}

\newtheorem{theorem}{Theorem}[section]
\newtheorem{proposition}[theorem]{Proposition}
\newtheorem{lemma}[theorem]{Lemma}

\newtheorem{remark}[theorem]{Remark}

\newcommand{\R}{{\mathbb R}}

\newcommand{\C}{{\mathcal{C}}}

\newcommand{\scs}{{\mathcal{S}}}

\newcommand{\K}{{\mathcal{K}}}

\newcommand{\F}{{\cal F}}

\newcommand{\X}{{\mathcal{X}}}

\newcounter{rcnt}[section]

\def\qt#1{\qquad\text{#1}}

\begin{document}

\title{Covering Numbers of $L_p$-balls of Convex Functions and Sets}  

\author{Adityanand Guntuboyina \thanks{Supported by NSF Grant
    DMS-1309356} \thanks{Adityanand Guntuboyina, 423 
    Evans Hall, Berkeley, CA - 94720. Email:
  \url{aditya@stat.berkeley.edu}} \\ Department of Statistics, University
  of California, Berkeley}  

\maketitle

\begin{abstract}
 We prove bounds for the covering numbers of classes of convex
 functions and convex sets in Euclidean space. Previous results
 require the underlying convex functions or sets to be uniformly
 bounded. We relax this assumption and replace it with weaker integral
 constraints. The existing results can be recovered as special cases
 of our results. 
\end{abstract}



\section{Introduction}
For a subset $\F$ of a space $\X$ equipped 
with a pseudometric $\rho$, the $\epsilon$-covering number $M(\F,
\epsilon; \rho)$ is defined as the smallest number of closed balls of
radius $\epsilon$ whose union contains $\F$. The quantity $\log M(\F,
\epsilon; \rho)$ is referred to as the $\epsilon$-metric
entropy. This notion was introduced by A.N. Kolmogorov and is also
referred to as the Kolmogorov $\epsilon$-entropy. Covering numbers and
metric entropy provide an important measure of the massivity of $\F$
and play a central role in a number of areas including approximation
theory, empirical processes, nonparametric function estimation and
statistical learning theory.  

In this paper, we study the covering numbers of classes of convex
functions and classes of convex sets in Euclidean space. For classes
of convex functions, the best existing results are due
to~\citet{Dryanov} (for $d = 1$) and~\citet{GS13} (for $d \geq
1$) who proved optimal upper and lower bounds for the covering numbers
of uniformly bounded convex functions under $L^q$ metrics for $1 \leq
q < \infty$ (the definition of $L^q$ metrics is recalled
in~\eqref{lqd}). Specifically, they considered the class
$\C_{\infty}(I, B)$ of all 
convex functions on $I := [a_1, b_1] \times \dots \times [a_d, b_d]$
which are uniformly   bounded by $B$ and proved optimal upper and
lower bounds (upto 
multiplicative constants) for $\log M(\C_{\infty}(I, B), \epsilon;
L_q(I))$ for $1 \leq q < \infty$. These results can be seen as an
improvement over the classical results of~\citet{Bronshtein76} who
considered convex functions that are uniformly Lipschitz in addition
to being uniformly bounded. It may be noted that the result
of~\cite{Dryanov} was originally motivated by a queston posed by
A.I. Shnirelman. 

A natural question regarding the results of~\citet{Dryanov}
and~\citet{GS13} is whether the uniform boundedness assumption is
necessary for obtaining $L^q$ covering numbers on classes of convex
functions. We address this question in this paper and we show that
uniform boundedness is not necessary and it can be replaced by an
$L^p$ constraint for any $p > q$. Specifically, we consider, for $1
\leq p < \infty$, the class $\C_p(I, B)$ of all convex functions on $I
:= [a_1, b_1] \times \dots \times [a_d, b_d]$ which satisfy the integral
constraint $\int_I |f(x)|^p dx \leq B^p$ and we prove the following
interesting phenomenon for $\log M(\C_p(I, B),
\epsilon; L_q(I))$: for $1 \leq q < p \leq \infty$, the metric entropy
is finite and is bounded from above and below by constant multiples of
$\epsilon^{-d/2}$ while for $1 \leq p \leq q \leq \infty$, the metric
entropy is infinite. The results of~\citet{Dryanov} and~\citet{GS13} can
therefore be seen as special cases of our results corresponding to the
case when $p = \infty$. 

We also prove that, for the case when $1 \leq p = q < \infty$, the
metric entropy is barely infinite in the following sense: for every
subrectangle $J := [\alpha_1, \beta_1] \times \dots \times [\alpha_d,
\beta_d]$ of $I$ with $a_i < \alpha_i < \beta_i < b_i$ for $i = 1,
\dots, d$, the metric entropy $\log M(\C_p(I, B), \epsilon; L_p(J))$
is bounded from above by $\epsilon^{-d/2}$ upto multiplicative factors
that are logarithmic in the lengths $\alpha_i - a_i$  and $b_i -
\beta_i$ for $i = 1, \dots, d$. 

We also consider classes of convex sets. Here the main existing result
on covering numbers is due to~\citet{Bronshtein76} who considered the
class $\K_{\infty}(R)$ of all compact convex subsets of $\R^d$ (for
$d \geq 2$) that are contained in the ball of radius $R$ centered at
the origin. Under the Hausdorff metric $\ell_H$
(the definition of the Hausdorff metric is recalled
in~\eqref{hdef}),~\citet{Bronshtein76} proved bounds for the metric
entropy of $\K_{\infty}(R)$. Specifically,~\citet{Bronshtein76} proved
that $\log M(\K_{\infty}(R), \epsilon; \ell_H)$ is bounded from above
and below by constant multiples of $\epsilon^{(1-d)/2}$. A similar but
weaker result is proved in~\citet{Dudley74}. 

Using the notion of support function, the class $\K_{\infty}(R)$ can
be thought of as an $L_{\infty}$-ball in the class of all compact,
convex subsets of $\R^d$. The support function $h_K$ of a compact,
convex subset $K$ of $\R^d$ ($d \geq 2$) is defined for $u$ in the
unit sphere, $S^{d-1} := \left\{x \in \R^d: x_1^2 + \dots +   x_d^2 =
  1 \right\}$, by  
\begin{equation*}
  h_K(u) := \sup_{x \in K} (x \cdot u) \qt{where $x \cdot u := x_1 u_1
    + \dots + x_d u_d$}. 
\end{equation*}
Elementary properties of the support function can be found
in~\cite[Section 1.7]{Schneider} or~\cite[Section
13]{Rockafellar70book}. Using the support function, the class
$\K_{\infty}(R)$ can be written as $\left\{K
  \in \K : \sup_{u \in S^{d-1}} |h_K(u)| \le R \right\}$ where
$\K$ is the class of all compact, convex subsets of $\R^d$. A 
natural question now is to ask for covering numbers of the classes:
\begin{equation}\label{kpd}
  \K_{p}(R) := \left\{K \in \K : \int_{S^{d-1}} |h_K(u)|^p d\nu(u)
  \leq R^p \right\} \qt{for $1 \leq p < \infty$}
\end{equation}
where $\nu$ is the uniform probability measure on $S^{d-1}$. These
classes are all larger than $\K_{\infty}(R)$. In Theorem~\ref{sr} of
this paper, we prove that, for every $1 \leq p \leq \infty$, the
metric entropy $\log M(\K_p(R), \epsilon; \ell_H)$ is bounded from
above and below by constant multiples of $\epsilon^{(1-d)/2}$. 

The rest of the paper is organized as follows. We state our results
for the metric entropy of classes of convex functions in
Section~\ref{cf}. We prove these results in Section~\ref{pcf}. The
main idea behind our convex function results can be isolated into a
separate theorem which we state and prove in
Section~\ref{divisec}. Our results for convex sets are stated and
proved in Section~\ref{cc}.  The proof of an auxiliary result is given
in Section~\ref{apx}. 

\section{Convex Functions}~\label{cf}
Recall the notions of $\C_p(I, B)$ for $1 \leq p \leq \infty$, $I =
[a_1, b_1] \times \dots \times [a_d, b_d]$ and $B > 0$. Also recall
that under the $L_q(J)$ metric on a subset $J$ of $\R^d$, the distance
between two functions $f$ and $g$ on $J$ is defined as   
\begin{equation}\label{lqd}
  \left(\int_J |f(x) - g(x)|^q dx \right)^{1/q} \qt{for $1 \leq q < \infty$}
\end{equation}
and as $\sup_{x \in J} |f(x) - g(x)|$ for $q = \infty$. 

Guntuboyina and Sen~\cite[Theorem 3.1]{GS13} proved the
following result for the metric entropy of $\C_{\infty}(I, B)$ under
the $L_q(I)$ metric for $1 \leq q < \infty$. Dryanov~\cite{Dryanov}
previously proved the special case of this result for $d = 1$. 
\begin{theorem}[Guntuboyina and Sen]\label{IEEE}
Fix $d \geq 1$ and $1 \leq q < \infty$. There exist positive
constants $c_1, c_2$ and $\epsilon_0$ depending only on $d$ and $q$
such that for every $B > 0$ and $I = [a_1, b_1] \times \dots \times
[a_d, b_d]$, we have  
  \begin{equation}\label{IEEE.eq}
    \log M \left(\C_{\infty}(I, B), \epsilon; L_q(I) \right) \leq c_1
    \left(\frac{\epsilon}{B (b_1 - a_1)^{1/q} \dots (b_d - a_d)^{1/q}}
    \right)^{-d/2} 
  \end{equation}
for all $\epsilon > 0$ and 
\begin{equation} \label{IEEE.eq1}
  \log M \left(\C_{\infty}(I, B), \epsilon; L_q(I) \right) \geq c_2
  \left(\frac{\epsilon}{B (b_1 - a_1)^{1/q} \dots (b_d - a_d)^{1/q}}
  \right)^{-d/2} 
\end{equation}
whenever $0 < \epsilon \leq \epsilon_0 B (b_1 - a_1)^{1/q} \dots (b_d
- a_d)^{1/q}$. 
\end{theorem}

\begin{remark}
  In~\cite{GS13}, the above result was proved only for rectangles
  of the form $[a, b]^d$ as opposed to $[a_1, b_1] \times \dots
  \times [a_d, b_d]$. But it is easy to see that the result for $[a,
  b]^d$ implies Theorem~\ref{IEEE} by a scaling argument (for example,
  via inequality~\eqref{sci}). 
\end{remark}

\begin{remark}
 In~\cite{GS13}, inequality~\eqref{IEEE.eq} was  only proved for
 $\epsilon \leq \epsilon_0 B(b_1 - a_1)^{1/q} \dots 
  (b_q - a_q)^{1/q}$ for a positive constant $\epsilon_0$ depending
  only on $d$ and $q$. It turns out however that this condition is 
  redundant. This follows from the observation that the diameter of
  the space $\C_{\infty}(I, B)$ in the
  $L_q(I)$ metric is at most $2B(b_1 - a_1)^{1/q} \dots (b_d -
  a_d)^{1/q}$ which means that left hand side of~\eqref{IEEE.eq}
  equals $0$ for $\epsilon > 2B(b_1 - a_1)^{1/q} \dots
  (b_d-a_d)^{1/q}$.
\end{remark}

In this paper, we extend Theorem~\ref{IEEE} by proving bounds for the
metric entropy of $\C_{p}(I, B)$ for $1 \leq p <
\infty$. Note that functions in $\C_p(I, B)$ for $1 \leq p < \infty$
do not have to be uniformly bounded as in $\C_{\infty}(I, 
B)$ but instead they are only required to satisfy a weaker integral
constraint. We prove the following result for the metric entropy of
these classes under $L_q$ metrics: for $q <
p$, the metric entropy under the $L_q$ metric is finite and is bounded
from above by a constant multiple of $\epsilon^{-d/2}$ while for $q
\geq p$, the metric entropy under the $L_q$ metric is infinite. The
fact that the metric entropy is infinite when $q \geq p$ is shown in
Theorem~\ref{fin} while bounds on the metric entropy for $q < p$ are
proved in Theorem~\ref{mth}. It is clear that Theorem~\ref{IEEE} is a
special case of Theorem~\ref{mth} corresponding to $p = \infty$. 

\begin{theorem}\label{mth}
  Fix $d \geq 1$ and $1 \leq q < p \leq \infty$. There exist positive
  constants $c_1, c_2$ and $\epsilon_0$ depending only on $d$, $p$ and
  $q$ such that  
  \begin{equation}\label{mth.eq1}
    \log M \left(\C_p(I, B), \epsilon; L_q \right) \leq c_1 \left(\frac{\epsilon}{B
        (b_1 - a_1)^{1/q - 1/p} \dots (b_d - a_d)^{1/q - 1/p}}
    \right)^{-d/2}  
  \end{equation}
for every $\epsilon > 0$ and 
\begin{equation}
  \label{mth.eq2}
   \log M \left(\C_p(I, B), \epsilon; L_q \right) \geq c_2
   \left(\frac{\epsilon}{B (b_1 - a_1)^{1/q - 1/p} \dots (b_d -
       a_d)^{1/q - 1/p}} 
    \right)^{-d/2}  
\end{equation}
whenever $0 < \epsilon \leq \epsilon_0 B (b_1 - a_1)^{1/q - 1/p} \dots
(b_d - a_d)^{1/q - 1/p}$. 
\end{theorem}

\begin{theorem}\label{fin}
  Fix $d \geq 1$ and $1 \leq p \leq q \leq \infty$. There exists a
  positive constant $\epsilon_0$ depending only on $d, p$ and $q$ such
  that 
  \begin{equation}\label{fin.eq}
    \log M(\C_p(I, B), \epsilon; L_q(I)) = \infty 
  \end{equation}
  whenever $\epsilon\leq B \epsilon_0 (b_1 - a_1)^{1/q - 1/p} \dots
  (b_d - a_d)^{1/q - 1/p}$. 
\end{theorem}
When $1 \leq p = q < \infty$, it turns out that $\log M(\C_p(I, B),
\epsilon; L_p(I))$ is barely infinite in the sense 
made precise by the theorem below. Note that the dependence on $\eta$
in the next theorem 
is logarithmic. 

\begin{theorem}\label{peq}
  Fix $d \geq 1$ and $1 \leq p < \infty$. There exists a positive
  constant $c$ depending only on $d$ and $p$ such that for every $I =
  [a_1, b_1] \times \dots \times [a_d, b_d]$ and $J = [\alpha_1,
  \beta_1] \times \dots \times [\alpha_d, \beta_d]$ with $a_i <
  \alpha_i < \beta_i < b_i$ with $i = 1, \dots, d$, we have 
  \begin{equation*}
    \log M(\C_p(I, B), \epsilon; L_p(J))  \leq  c
    \left(\frac{\epsilon}{B}
    \right)^{-d/2} \left(\log \frac{1}{\eta} \right)^{d(2p+d)/(2p)} 
  \end{equation*}
  for all $\epsilon > 0$ where 
  \begin{equation}\label{edf}
    \eta :=\min \left(\frac{\alpha_1 -
        a_1}{b_1 - a_1}, \dots, \frac{\alpha_d - a_d}{b_d - a_d},
      \frac{b_1 - \beta_1}{b_1 - a_1}, \dots, \frac{b_d -
        \beta_d}{b_d - a_d}  \right). 
  \end{equation}
\end{theorem}

The main idea behind the proofs of Theorem~\ref{mth} and~\ref{peq} is
the following: scaling identities~\eqref{sci} and~\eqref{sci2}
described in Section~\ref{pcf} allow us to take $I = [0, 1]^d$. We
then show that functions in $\C_p([0, 1]^d, 1)$ are uniformly bounded
on subrectrangles that are contained in the interior of $[0, 1]^d$. We
divide $[0, 1]^d$ into such subrectangles and apply
Theorem~\ref{IEEE} in each of the subrectangles. The proofs are then
completed by combining the metric entropy bounds from
Theorem~\ref{IEEE} for each of the
subrectangles. This method, which we call the partitioning method, can be
isolated into a separate theorem (Theorem~\ref{divis}) which we state
and prove in the next section. The partitioning method is a
multivariate extension and simplification of the multistep
approximation procedure given in~\cite[Theorem 3.1]{Dryanov}. In
Section~\ref{pcf}, we show how Theorems~\ref{mth}
and~\ref{peq} can be proved from Theorem~\ref{divis}. We also provide
the proof of Theorem~\ref{fin} in Section~\ref{pcf}. 

\section{The Partitioning Theorem}\label{divisec}
\begin{theorem}\label{divis}
  Fix $1 \leq p < \infty$ and $1 \leq q < \infty$. There exists a
  constant $c$ depending only on $d$, $p$ and $q$ such that the
  following inequality is true for every $0 < \eta < u \leq 1/2$,  $l
  \geq 1$ and every finite sequence $\eta = \eta_0 < \eta_1 < \dots <
  \eta_l < u \leq \eta_{l+1}$:  
  \begin{equation*}
    \log M(\C_p([0, 1]^d, 1), \epsilon, L_q[\eta, u]^d) \leq c
    \epsilon^{-d/2} \left(\sum_{i=0}^l \frac{(\eta_{i+1} -
        \eta_i)^{d/(2q+d)}}{\eta_i^{dq/(p(2q+d))}}
    \right)^{d(2q+d)/(2q)}
  \end{equation*}
for all $\epsilon > 0$. 
\end{theorem}
We need two preparatory results for the proof of
Theorem~\ref{divis}. The first of these results is given below. Its
proof is trivial and is omitted.   
\begin{lemma}\label{rul}
  Let $\F$ be an arbitrary class of functions defined on a subset $A$
  of $\R^d$ and let $A_1 \dots, A_k$ denote subsets of $\R^d$ with
  $A \subseteq \cup_{i=1}^k A_i$. Then for every $\epsilon,
  \epsilon_1, \dots,  \epsilon_k > 0$, we have  
  \begin{equation*}
    \log M(\F, \epsilon; L_q(A)) \leq \sum_{i=1}^k \log M(\F,
    \epsilon_i, L_q(A_i)) 
  \end{equation*}
provided $\sum_{i=1}^k \epsilon_i^q \leq \epsilon^q$. 
\end{lemma}
The second preparatory result states that for every $\phi \in \C_p([0,
1]^d, 1)$ and $y \in (0, 1)^d$, the quantity $|\phi(y)|$ can be
bounded from above by a term that is independent of $\phi$. The
precise statement is given below and its proof is deferred to
Section~\ref{apx}.  
. 
\begin{lemma}\label{gbr}
Let $1 \leq p \leq \infty$ and let $\phi$ be a convex function on
$[0,1]^d$ with $\int_{[0, 1]^d} |\phi(x)|^p dx \leq 1$. Then there
exists a positive constant $c$ depending only on $d$ and $p$ such that 
for every $y = (y_1, \dots, y_d) \in (0, 1)^d$,
\begin{equation}\label{gbr.eq}
  |\phi(y)| \leq c \prod_{i=1}^d \max
  \left(y_i^{-1/p}, (1 -  y_i)^{-1/p} \right). 
\end{equation}
\end{lemma}
We are now ready to prove Theorem~\ref{divis}. 
\begin{proof}[Proof of Theorem~\ref{divis}]
  Let us fix $0 < \eta < u \leq 1/2$ and an arbitrary finite sequence
$\eta = \eta_0 < \eta_1 < \dots < \eta_l < u \leq \eta_{l+1}$ for a
positive integer $l \geq 1$. For every $f$ and $g$, we can write 
  \begin{equation*}
    \int_{[\eta, u]^d} \left|f(x) - g(x) \right|^q dx \leq 
    \sum_{i_1=0}^{l} \dots \sum_{i_d = 0}^l
    \int_{\eta_{i_1}}^{\eta_{i_1+1}} \dots
    \int_{\eta_{i_d}}^{\eta_{i_d+1}}  
    \left|f(x) - g(x) \right|^q dx_1 \dots dx_d. 
  \end{equation*}
Lemma~\ref{gbr} asserts that every function $\phi \in \C_p([0,
1]^d, 1)$, when restricted to the rectangle $[\eta_{i_1}, \eta_{i_1+1}]  \times
\dots \times [\eta_{i_d}, \eta_{i_d+1}]$, is convex and uniformly
bounded by $C \eta_{i_1}^{-1/p} \dots \eta_{i_d}^{-1/p}$ for a
constant $C$ that only depends on $d$ and $p$. Therefore,
by Theorem~\ref{IEEE}, we can cover the restrictions of functions in
$\C_p([0, 1]^d, 1)$ to $[\eta_{i_1}, \eta_{i_1+1}]  \times \dots \times
[\eta_{i_d}, \eta_{i_d+1}]$ to within a positive real number
$\alpha(i_1, \dots, i_d)$ in the $L_q$ metric on $[\eta_{i_1},
\eta_{i_1+1}]  \times \dots \times [\eta_{i_d}, \eta_{i_d+1}]$  
by a finite set having cardinality at most  
\begin{equation*}
  \exp \left(c \left(\frac{\alpha(i_1, \dots, i_d) \eta_{i_1}^{1/p}
        \dots \eta_{i_d}^{1/p}}{(\eta_{i_1+1} - \eta_{i_1})^{1/q}
        \dots (\eta_{i_d+1} - \eta_{i_d})^{1/q}} \right)^{-d/2}
  \right)
\end{equation*}
where $c$ is a positive constant that only depends on $d$, $p$ and
$q$. By Lemma~\ref{rul} therefore, we get an $\epsilon$-cover for
functions in $\C_p([0, 1]^d, 1)$ in the $L_q$ metric on $[\eta, u]^d$
with   
\begin{equation*}
  \epsilon^q = \sum_{i_1 = 0}^l \dots \sum_{i_d = 0}^l
  \alpha^q(i_1, \dots, i_d)
\end{equation*}
having cardinality at most
\begin{equation}\label{beem}
  \exp \left[c \sum_{i_1=0}^l \dots \sum_{i_d = 0}^l
    \left(\frac{\alpha(i_1, \dots, i_d) \eta_{i_1}^{1/p} \dots
        \eta_{i_d}^{1/p}}{(\eta_{i_1+1} - \eta_{i_1})^{1/q} \dots
        (\eta_{i_d+1} - \eta_{i_d})^{1/q}} \right)^{-d/2} \right]. 
\end{equation}
For each $i := (i_1, \dots, i_d) \in \{0, \dots, l\}^{d}$, let  
\begin{equation*}
  u_i  := \frac{\eta_{i_1}^{1/p} \dots
    \eta_{i_d}^{1/p}}{(\eta_{i_1 + 1} - \eta_{i_1})^{1/q} \dots
    (\eta_{i_d+1} - \eta_{i_d})^{1/q}}.  
\end{equation*}  
Plugging in the choice
\begin{equation*}
  \alpha(i_1, \dots, i_d) := \epsilon u_i^{-d/(d+2q)} \left(\sum_{i
      \in \{0, \dots, l\}^d} u_i^{-dq/(d+2q)} \right)^{-1/q} 
\end{equation*}
into~\eqref{beem}, we obtain that 
\begin{equation*}
  \log M \left(\epsilon, \C_p([0, 1]^d, 1), L_q[\eta, u]^d \right)
  \leq c \epsilon^{-d/2} \left(\sum_{i} u_i^{-dq/(2q+d)}
  \right)^{(2q+d)/(2q)} . 
\end{equation*}
The observation
\begin{equation*}
  \sum_{i \in \{0, \dots, l\}^d} u_i^{-dq/(2q+d)} = \left(\sum_{i=0}^l
    \frac{(\eta_{i + 1} - \eta_i)^{d/(2q+d)}}{\eta_i^{dq/(p(2q+d))}} \right)^d . 
\end{equation*}
now completes the proof.
\end{proof}

\section{Proofs for results in Section~\ref{cf}}\label{pcf}
We give the proofs of Theorems~\ref{mth},~\ref{fin} and~\ref{peq} in
this section. We start with a pair of simple scaling identities which
allow us to take $I = [0, 1]^d$ without loss of generality. The first
identity is: For every $I = [a_1, b_1] \times \dots \times [a_d,
b_d]$, we have
\begin{equation}\label{sci}
  M(\C_p(I, B), \epsilon; L_q(I)) = M(\C_p([0, 1]^d, 1),
  \tilde{\epsilon}, L_q([0, 1]^d))
\end{equation}
where 
\begin{equation*}
  \tilde{\epsilon} := (b_1 - a_1)^{1/p - 1/q} \dots (b_d - a_d)^{1/p -
  1/q} \frac{\epsilon}{B}. 
\end{equation*}
To see~\eqref{sci}, associate for each $f \in \C_p(I, B)$, the
function $\tilde{f}$ on $[0, 1]^d$ by  
   \begin{equation*}
     \tilde{f}(x_1, \dots, x_d) := B^{-1} (b_1 - a_1)^{1/p} \dots (b_d
     - a_d)^{1/p} f(a_1 + (b_1 - a_1) x_1, \dots, a_d 
     + (b_d - a_d)x_d). 
   \end{equation*}
   It is then easy to verify that $\tilde{f}$ lies in $\C_{p}([0, 1]^d,
   1)$ and that covering $\tilde{f}$ to within $\tilde{\epsilon}$ in
   the $L_{q}$ metric on $[0, 1]^d$ is equivalent to covering $f$ to
   within $\epsilon$ in the $L_q$ metric on $I$ and this
   proves~\eqref{sci}. The identity~\eqref{sci} implies that we can,
   without loss of generality, take $I = [0, 1]^d$ and $B = 1$ in the
   proofs of Theorems~\ref{mth} and~\ref{fin}. 

  The second scaling identity is: For every $I = [a_1, b_1] \times
  \dots \times [a_d, b_d]$ and $J = [\alpha_1, \beta_1] \times \dots
  \times [\alpha_d, \beta_d]$ with $a_i < \alpha_i < \beta_i < b_i$
  for all $i$, we have
  \begin{equation}\label{sci2}
    M(\C_p(I, B), \epsilon, L_p(J)) \leq M(\C_p([0, 1]^d, 1),
    \epsilon/B, L_p[\eta, 1- \eta]^d) 
  \end{equation}
  where $\eta$ is defined as in~\eqref{edf}. The proof of~\eqref{sci2}
  is similar to that of~\eqref{sci} and is thus
  omitted. Identity~\eqref{sci2} allows us to take, without loss of
  generality, $I = [0, 1]^d$, $J = [\eta, 1-\eta]^d$ and $B = 1$ for
  the proof of Theorem~\ref{peq}. 

\subsection{Proof of Theorem~\ref{mth}}
Inequality~\eqref{mth.eq2} is a direct consequence of~\eqref{IEEE.eq1}
because
\begin{equation*}
  \C_{\infty}(I, B (b_1 -
  a_1)^{-1/p} \dots (b_d - a_d)^{-1/p}) \subseteq \C_p(I, B) \qt{for
    every $1 \leq p \leq \infty$}. 
\end{equation*}
We therefore only need to prove~\eqref{mth.eq1}.  We assume that $p <
\infty$ because the case when $p = \infty$ is taken care of by
Theorem~\ref{IEEE}.  The scaling inequality~\eqref{sci} allows us to
restrict attention to the case when $I = [0, 1]^d$ and $B =
1$. Therefore, we only need to prove the existence of a positive  
constant $c$ (depending only on $d, p$ and $q$) such that 
\begin{equation}\label{rk}
     \log M(\C_p([0, 1]^d, 1), \epsilon; L_q[0, 1]^d) \leq c \epsilon^{-d/2}
     \qt{for all $\epsilon > 0$}. 
\end{equation}

 Our first step for the proof of~\eqref{rk} is to reduce focus to the 
 $L_q$ metric on a subrectangle $[\eta, 1/2]^d$ of $[0, 1]^d$ for some
 $\eta > 0$ as opposed to the $L_q$ metric on the entire unit cube $[0,
 1]^d$.  

\subsubsection{Reduction to the $L_q[\eta, 1/2]^d$ metric for $0 < \eta < 1/2$} 
The behaviour of functions in $\C_p([0, 1]^d, 1)$ can be difficult to
control near the boundary of the cube $[0, 1]^d$. For this reason,
the metric entropy of $\C_p([0, 1]^d, 1)$ under the pseudo-metric
$L_q[\eta, 1-\eta]^d$ for $0 < \eta< 1/2$ will be easier to bound than
the metric entropy under $L_q[0, 1]^d$. The following lemma proves
that it is actually enough to work with $L_q[\eta, 1-\eta]^d$ for some
appropriately chosen $\eta$, $0 < \eta < 1/2$ depending on $\epsilon$. 
\begin{lemma}\label{wlg1}
Let 
\begin{equation*}
0 <\epsilon < 2^{1/q} d^{(1/q) - (1/p)} ~~ \text{ and } ~~
\eta_{\epsilon} := \frac{1}{2d} \left(\frac{\epsilon^q}{2}
\right)^{p/(p-q)}. 
\end{equation*}
Then 
\begin{equation}\label{wlg1.eq}
  M \left(\C_p([0, 1]^d, 1), \epsilon ; L_q[0, 1]^d  \right) \leq M
  \left(\C_p([0,1]^d, 1), \epsilon 2^{-1/q}; L_q[\eta_{\epsilon}, 1 -
    \eta_{\epsilon}]^d \right). 
\end{equation}
\end{lemma}
\begin{proof}
Fix $\phi \in \C_p([0, 1]^d, 1)$  and an arbitrary function $\psi$ on
$[\eta_{\epsilon}, 1 - \eta_{\epsilon}]^d$ where $\eta_{\epsilon}$ is
defined as in the statement of the lemma. Extend $\psi$ to $[0, 1]^d$
by defining it to be zero outside $[\eta_{\epsilon}, 1 -
\eta_{\epsilon}]^d$. Observe that
\begin{equation}\label{wi}
  \int_{[0, 1]^d} |\phi - \psi|^q = \int_{[\eta_{\epsilon}, 1 -
    \eta_{\epsilon}]^d} |\phi - \psi|^q + \int_{[0, 1]^d} |\phi(x)|^q
  I\left\{x \notin [\eta_{\epsilon}, 1 - \eta_{\epsilon}]^d \right\} dx. 
\end{equation}
Applying Holder's inequality $\int |fg| \leq (\int |f|^r)^{1/r} (\int
|g|^s)^{1/s}$ with $f := |\phi|^q$, $g := I\{x \notin [\eta_{\epsilon},
1 - \eta_{\epsilon}]^d\}$, $r = p/q$ and $s = p/(p-q)$, we get 
\begin{align*}
  \int_{[0, 1]^d}  |\phi(x)|^q I \left\{x \notin [\eta_{\epsilon}, 1 -
    \eta_{\epsilon}]^d \right\} dx &\leq \left(\int_{[0, 1]^d}
    |\phi(x)|^p \right)^{q/p} \left(1 - (1 - 2 \eta_{\epsilon})^d
  \right)^{1-(q/p)} \\
&\leq \left(\int_{[0, 1]^d} |\phi(x)|^p \right)^{q/p} \left( 2 d
  \eta_{\epsilon} \right)^{1-(q/p)} \leq \frac{\epsilon^q}{2}. 
\end{align*}
This, together with~\eqref{wi}, gives 
\begin{equation*}
  \int_{[0, 1]^d} |\phi - \psi|^q \leq  \int_{[\eta_{\epsilon}, 1 -
    \eta_{\epsilon}]^d} |\phi - \psi|^q + \frac{\epsilon^q}{2}
\end{equation*}
from which~\eqref{wlg1.eq} follows immediately. 
\end{proof}
By symmetry, it can be shown that the metric entropy of $\C_p([0,
1]^d, 1)$ under $L_q[\eta, 1-\eta]^d$ is bounded from above by a
constant multiple of the metric entropy under $L_q[\eta, 1/2]^d$. This
is the content of the following lemma. 
\begin{lemma}\label{wlg2}
The following inequality holds for every $0 < \eta < 1/2$ 
\begin{equation}
  \label{wlg2.eq}
 \log M(\C_p([0, 1]^d, 1), \epsilon; L_q[\eta, 1-\eta]^d) \leq 2^d 
\log M(\C_p([0, 1]^d, 1), \epsilon 2^{-d/q}, L_q[\eta, 1/2]^d). 
\end{equation}
\end{lemma}
\begin{proof}
  Let $I(0) := [\eta, 1/2]$ and $I(1) := [1/2, 1- \eta]$. For any pair
  of functions $\phi$ and $\psi$, observe that 
  \begin{equation*}
    \int_{[\eta, 1-\eta]^d} |\phi - \psi|^q = \sum_{\theta \in \{0,
      1\}^d} \int_{I(\theta_1) \times \dots \times I(\theta_d)} |\phi - \psi|^q 
  \end{equation*}
which implies, by Lemma~\ref{rul}, that 
\begin{equation*}
  M(\C_p([0, 1]^d, 1), \epsilon; L_q[\eta, 1-\eta]^d) \leq
  \prod_{\theta \in \{0, 1\}^d} M(\C_p([0, 1]^d, 1), \epsilon
  2^{-d/q}; L_q(I(\theta_1) \times \dots \times I(\theta_d))) . 
\end{equation*}
By symmetry, we get that
\begin{equation*}
  M(\C_p([0, 1]^d, 1), \epsilon
  2^{-d/q}; L_q(I(\theta_1) \times \dots \times I(\theta_d))) =
  M(\C_p([0, 1]^d, 1), \epsilon 2^{-d/q}, L_q[\eta, 1/2]^d)
\end{equation*}
for every $\theta \in \{0, 1\}^d$. This completes the proof
of~\eqref{wlg2.eq}.  
\end{proof}
The above pair of results (Lemma~\ref{wlg1} and~\ref{wlg2}) together
imply that~\eqref{rk} will be a consequence of the
following result: 
\begin{proposition}\label{kred}
  Fix $d \geq 1$ and $1 \leq q < p < \infty$. There exists positive
  constants $c$ and $\epsilon_0$ depending only on $d$, $p$ and $q$
  such that 
  \begin{equation*}
  \sup_{0 < \eta < 1/2}  \log M(\C_p([0, 1]^d, 1), \epsilon; L_q[\eta,
  1/2]^d) \leq c \epsilon^{-d/2}   \qt{for every $\epsilon \leq \epsilon_0$}.
\end{equation*}
\end{proposition}

Proposition~\ref{kred} will be proved in the next
subsection. This will complete the proof of Theorem~\ref{mth} . 

\subsubsection{Proof of Proposition~\ref{kred}}\label{kpred}
Fix $p > q$ and let  
\begin{equation}\label{expu}
  u :=  \exp \left(\frac{-2p(p+q) (2q + d) \log 2}{d(p- q)^2} \right). 
\end{equation}
Note that $u$ only depends on $p, q$ and $d$ and that $0 < u <
1/2$. 

Using the notation $a \vee b := \max(a, b)$, we can write  
\begin{equation*}
  \int_{[\eta, 1/2]^d} |f - g|^q dx = \int_{[\eta, u \vee
    \eta]^d} |f - g|^q + \int_{[u \vee \eta, 1/2]^d} |f -
  g|^q \leq \int_{[\eta, u \vee
    \eta]^d} |f - g|^q + \int_{[u, 1/2]^d} |f -
  g|^q . 
\end{equation*}
for every pair of functions $f$ and $g$. Applying Lemma~\ref{rul},
we obtain
\begin{equation*}
  M(\C_p([0, 1]^d, 1), \epsilon; L_q[\eta, 1/2]^d) \leq M(\C_p([0,
  1]^d, 1), 2^{-1/q}\epsilon; L_q[\eta, u \vee \eta]^d) M(\C_p([0,
  1]^d, 1), 2^{-1/q}\epsilon; L_q[u, 1/2]^d). 
\end{equation*}
Now, by Lemma~\ref{gbr}, every function in $\C_p([0, 1]^d, 1)$, when
restricted to $[u, 1/2]^d$ is convex and uniformly bounded by $C
u^{-d/p}$ for a positive constant $C$ that only depends on $d$ and
$p$. It follows from Theorem~\ref{IEEE} (and the fact that $u$ is a
constant that only depends on $d, p$ and $q$) that there exists a
constant $c$ depending on $d, p$ and $q$ alone such
that  
\begin{equation*}
 \log M(\C_p([0, 1]^d, 1), 2^{-1/q}\epsilon; L_q[u, 1/2]^d) \leq c
 \epsilon^{-d/2} \qt{for all $\epsilon > 0$}. 
\end{equation*}
We deduce therefore that the proof of Proposition~\ref{kred} will be
complete if we prove the existence of a constant $c$ for which 
\begin{equation}\label{fred}
  \log M(\C_p([0, 1]^d, 1), \epsilon; L_q[\eta, u \vee
  \eta]^d) \leq c \epsilon^{-d/2}. 
\end{equation}
We prove~\eqref{fred} below. It is trivial when $\eta \geq u$ so we
assume below that $\eta < u$. By Theorem~\ref{divis}, there exists a
positive constant $c$ depending only on $d$, $p$ and $q$ such that 
\begin{equation}\label{wr}
  \log M(\C_p([0, 1]^d, 1), \epsilon, L_q[\eta, u]^d) \leq c
  \epsilon^{-d/2} \left(\sum_{i=0}^l \frac{(\eta_{i+1} -
      \eta_i)^{d/(2q+d)}}{\eta_i^{dq/(p(2q+d))}}
  \right)^{d(2q+d)/(2q)} . 
\end{equation}
We use this with 
\begin{equation*}
  \eta_i := \exp \left( \left(\frac{p+q}{2p} \right)^i \log \eta
  \right) \qt{for $i \geq 1$}
\end{equation*}
and $l$ taken to be the largest integer $i$ for which $\eta_i <
u$. Because $p > q$ and $\log \eta < 0$,  it is clear that
$\{\eta_i\}$ is an increasing sequence. 

We shall show below that for this choice of $l$ and $\{\eta_i\}$, 
\begin{equation*}
  S := \sum_{i=0}^l \frac{(\eta_{i+1} -
    \eta_i)^{d/(2q+d)}}{\eta_i^{dq/(p(2q + d))}} \leq C
\end{equation*}
for a positive constant $C$ that only depends on $d$, $p$ and
$q$. The proof would then be complete by~\eqref{wr}. 

Define 
\begin{equation*}
  \zeta_i := \frac{\eta_{i+1}^{d/(2q+d)}}{\eta_i^{dq/(p(2q+d))}} = \exp
  \left(\frac{d(p-q)}{2p(2q+d)} \left(\frac{p+q}{2p} \right)^i \log
    \eta \right). 
\end{equation*}
Observe that for $1 \leq i \leq l$
\begin{align*}
  \frac{\zeta_i}{\zeta_{i-1}} &= \exp \left(\frac{-d
      (p-q)^2}{4p^2(2q+d)} \left(\frac{p+q}{2p} \right)^{i-1} \log
    \eta \right) \\
&= \exp \left(\frac{-d
      (p-q)^2}{2p(p+q)(2q+d)} \log \eta_i \right) \geq \exp \left(\frac{-d
      (p-q)^2}{2p(p+q)(2q+d)} \log u \right) = 2
\end{align*}
where we have used $\eta_i < u$ for $1 \leq i \leq l$ and the
expression~\eqref{expu} for $u$. This means that $\zeta_i \leq
2(\zeta_i-\zeta_{i-1})$ for $i = 1, \dots, l$ and, as a result, we get 
\begin{equation*}
  S \leq \sum_{i=0}^l \zeta_i = \zeta_0 + 2 \sum_{i=1}^l (\zeta_i -
  \zeta_{i-1}) = 2 \zeta_l - \zeta_0 \leq 2 \zeta_l.  
\end{equation*}
$\zeta_l$ can be bounded by a constant independent of $\eta$ because
\begin{equation*}
  \zeta_l = \exp \left(\frac{d(p-q)}{2p(2q+d)} \left(\frac{p+q}{2p}
    \right)^l \log \eta \right) =\exp \left(\frac{d(p-q)}{2p(2q+d)}
    \log \eta_l \right) < \exp \left(\frac{d(p-q)}{2p(2q+d)} \log u
  \right). 
\end{equation*}
This proves that $S$ is bounded from above by a constant that only
depends on $d$, $p$ and $q$. This  completes the proof of
Proposition~\ref{kred} and thereby that of Theorem~\ref{mth}. 

\subsection{Proof of Theorem~\ref{fin}}

   Because of~\eqref{sci}, it is sufficient to prove the theorem for
   $a_1 = \dots = a_d = 0$, $b_1 = \dots = b_d = 1$ and $B =
   1$. Define, for $j \geq 1$,   
  \begin{equation*}
    f_j(x) := (1 + p)^{1/p} 2^{j/p} \max(0, 1 - 2^j x_1) \qt{for $x =
      (x_1, \dots, x_d) \in [0, 1]^d$}. 
  \end{equation*}
  It is easy to check that $f_j \in \C_p([0, 1]^d, 1)$ for every $j
  \geq 1$. Now for $j < k$, note that 
  \begin{equation*}
 (1 + p)^{-q/p} \int_{[0, 1]^d} |f_j(x) - f_k(x)|^q dx \geq \int_{2^{-k}}^{2^{-j}} 
    2^{jq/p} (1 - 2^j x_1)^q dx_1= \frac{2^{j(q-p)/p}}{q+1} \left(1 -
      2^{j-k} \right)^{q+1} \geq c 
  \end{equation*}
 for some positive constant $c$ that depends only on $d, p$ and
 $q$. We have thus shown that for every pair of distinct functions
 from the infinite sequence $\{f_j\}$, the $L_q$ distance between them
 is bounded from below by a positive constant that only depends on $p$
 and $q$. This proves the theorem. 

\subsection{Proof of Theorem~\ref{peq}}
By the second scaling identity~\eqref{sci2}, we can take $I = [0,
1]^d$, $J = [\eta, 1-\eta]^d$ and $B = 1$ without loss of
generality. Further, by Lemma~\ref{wlg2}, it is enough to bound $\log
M(\C_p([0, 1]^d, 1), \epsilon, L_p[\eta, 1/2]^d)$. 

Theorem~\ref{divis} with $p = q$ and $u = 1/2$ gives the existence of
a constant $c$ for which
\begin{equation}\label{kk}
  \log M(\C_p([0, 1]^d, 1), \epsilon, L_p[\eta, 1/2]^d) \leq c
  \epsilon^{-d/2} \left(\sum_{i=0}^l \frac{(\eta_{i+1} -
      \eta_i)^{d/(2p +d)}}{\eta_i^{d/(2p+d)}} \right)^{d(2p+d)/(2p)} 
\end{equation}
for every $l \geq 1$ and every $\eta = \eta_0 < \eta_1 < \dots <
\eta_l < 1/2 \leq \eta_{l+1}$. We apply this to $\eta_i = 2^i \eta$
for $i = 0, \dots, l$ where $l := \lfloor -\log (2 \eta)/\log 2 \rfloor$ 
and $\eta_{l+1} = 1/2$. Here $\lfloor x \rfloor$ is the largest
integer that is strictly smaller than $x$. It is clear then that
$\eta_{i+1} - \eta_i \leq \eta_i$ and then, using~\eqref{kk}, we get 
\begin{equation*}
  \log M(\C_p([0, 1]^d, 1), \epsilon, L_p[\eta, 1/2]^d) \leq c
  \epsilon^{-d/2} \left(l+1\right)^{d(2p+d)/(2p)}. 
\end{equation*}
Because
\begin{equation*}
  1 + l \leq 1 + \frac{\log (1/(2\eta))}{\log 2} = \frac{\log
    (1/\eta)}{\log 2}, 
\end{equation*}
 we deduce that
\begin{equation*}
  \log M(\C_p([0, 1]^d, 1), \epsilon, L_p[\eta, 1/2]^d) \leq c_1
  \epsilon^{-d/2} \left(\log \frac{1}{\eta} \right)^{d(2p+d)/(2p)}
\end{equation*}
where $c_1$ only depends on $d$, $p$ and $q$. This completes the proof
of Theorem~\ref{peq}. 

\section{Convex sets}\label{cc}
Recall the definition of the classes $\K_p(R)$ for $1 \leq p \leq
\infty$ and $d \geq 1$ from~\eqref{kpd}. Bronshtein~\cite[Theorem 3
and Remark 1]{Bronshtein76} proved the following bound on the covering
number of $\K_{\infty}(R)$ under the Hausdorff metric $\ell_H$. It may
be recalled that the Hausdorff distance between two compact, convex
sets $C$ and $D$ in Euclidean space is defined by 
\begin{equation}\label{hdef}
  \ell_H(C, D) := \max \left(\sup_{x \in C} \inf_{y \in D} |x - y|,
    \sup_{x \in D} \inf_{y \in C} |x - y| \right)
\end{equation}
where $|\cdot|$ denotes Euclidean distance. 
\begin{theorem}[Bronshtein]\label{bsh}
There exist positive constants $c_1$ and $c_2$ depending only on $d$
such that for every $R > 0$, 
  \begin{equation}\label{bsh.eq}
c_1 \left(\frac{R}{\epsilon} \right)^{(d-1)/2} \leq  \log
M(\K_{\infty}(R), \epsilon; \ell_H) \leq c_2 \left(\frac{R}{\epsilon}
\right)^{(d-1)/2}.  
  \end{equation}
\end{theorem}
In the next theorem, we show that the same result~\eqref{bsh.eq} holds
for the covering number of $\K_p(R)$ for every $1 \leq p \leq
\infty$. To the best of our knowledge, this result is new. Note that
the classes $\K_p(R)$ for $1 \leq p < \infty$ are all larger than
$\K_{\infty}(R)$. 
\begin{theorem}\label{sr}
There exist positive constants $c_1$ and $c_2$ depending only on $d$
such that for every $1 \leq p \leq \infty$ and $R > 0$, 
  \begin{equation*}
c_1 \left(\frac{R}{\epsilon} \right)^{(d-1)/2} \leq \log M(\K_{p}(R),
\epsilon; \ell_H) \leq  c_2 \left(\frac{R}{\epsilon} \right)^{(d-1)/2}. 
  \end{equation*}   
\end{theorem}
\begin{proof}
$p = \infty$ corresponds to Theorem~\ref{bsh} so we may assume that $1
\leq p < \infty$. Because $\K_{\infty}(R) \subseteq \K_{p}(R)$, the
lower bound on $M(\K_p(R), \epsilon; \ell_H)$ follows from
Theorem~\ref{bsh}. We therefore only need to prove the upper bound. 
We show below that there exists a positive constant $M$ depending
only on $d$ and $p$ such that 
\begin{equation}\label{aap}
\K_p(R) \subseteq \K_{\infty}(M R)
\end{equation}
This means that $M(\K_p(R), \epsilon; \ell_H) \leq M(\K_{\infty}(MR),
\epsilon; \ell_H)$. The proof will then be complete by the use of
Theorem~\ref{bsh}. 
 
For each $v \in S^{d-1}$, define the spherical cap $\scs(v) :=
\left\{x \in S^{d-1}: ||x - v||^2 \leq 1 \right\}$. It is easy to
check that $\scs(v)$ can also be written as $\left\{x \in S^{d-1} : x
  \cdot v \geq 1/2 \right\}$.   

To prove~\eqref{aap}, fix $K \in \K_p(R)$ and $x \in K$. We need to
show that $x \in M R$ for a constant $M$ which only depends on $d$ and 
$p$. We may clearly assume that $x \neq 0$ and let $v :=
x/||x||$.  Note that for every $u \in S^{d-1}$, we have $h_K(u) \geq x
\cdot u = ||x|| (v \cdot u)$. Consequently, $h_K(u) \geq ||x||/2$
whenever $u \in S(v)$. As a result, 
\begin{equation*}
  R^p \geq \int_{S^{d-1}} \left|h_K(u) \right|^p d\nu(u) \geq
  \int_{\scs(v)} \left|h_K(u) \right|^p d\nu(u) \geq 2^{-p} ||x||^p
  \nu \left(\scs(v) \right)
\end{equation*}
which implies that $||x|| \leq 2 \nu(\scs(v))^{-1/p} R$. The quantity
$\nu(\scs(v))$ only depends on the dimension $d$ which completes the
proof. 
\end{proof}

\section{Appendix: Proof of Lemma~\ref{gbr}}\label{apx}
In this section, we provide the proof of Lemma~\ref{gbr} which was
crucially used in the proof of Theorem~\ref{divis}. Before we get to
the proof of Theorem~\ref{divis}, let us first state and prove a
technical result which we then use to prove Theorem~\ref{divis}. 
\begin{lemma}\label{fidi}
  Suppose $f$ is a continuous convex function on $[0, a]$ with $f(0)
  < 0$. Then for every $\alpha > 0$ and $p > 0$, we have 
  \begin{equation*}
    \int_0^{a} x^{\alpha - 1} |f(x)|^p dx \geq C(\alpha, p)|f(0)|^p
    a^{\alpha} 
  \end{equation*}
  where $C(\alpha, p)$ is the positive constant given by 
  \begin{equation}\label{fidi.c}
    C(\alpha, p) := \inf_{0 \leq \beta \leq 1} \int_0^1 u^{\alpha - 1}
    |u - \beta|^p du. 
  \end{equation}
\end{lemma}
\begin{proof}
  Suppose first that $f(a) \leq 0$. By convexity, we have
  \begin{equation*}
    f(x) \leq \frac{x}{a} f(a) + \left( 1 - \frac{x}{a} \right) f(0)
    \leq \left(1 - \frac{x}{a} \right) f(0)
  \end{equation*}
  and so we have $|f(x)| \geq (1 - (x/a)) |f(0)|$. As a consequence, 
  \begin{equation}\label{sc}
    \int_0^a x^{\alpha - 1} |f(x)|^p dx \geq |f(0)|^p \int_0^a x^{\alpha - 1}
    \left(1 - \frac{x}{a} \right)^p dx. 
  \end{equation}
  Now let $f(a) > 0$. By continuity, there exists $\beta \in (0, 1)$
  with $f(a\beta) = 0$. For $0 \leq x \leq a\beta$, we have by
  convexity 
  \begin{equation*}
    f(x) \leq \frac{x}{a\beta} f(a\beta) + \left( 1 - \frac{x}{a\beta}
    \right) f(0) = \left(1 - \frac{x}{a\beta} \right) f(0)
  \end{equation*}
which implies that
\begin{equation}\label{k2}
  \int_0^{a\beta} x^{\alpha - 1} |f(x)|^p dx \geq |f(0)|^p \int_0^{a\beta}
  x^{\alpha - 1} \left|1 - \frac{x}{a\beta} \right|^p dx. 
\end{equation}
On the other hand, for $a\beta \leq x \leq a$, we have, again by
convexity, 
\begin{equation*}
0 = f(a\beta) \leq \frac{a\beta}{x} f(x) + (1 - \frac{a\beta}{x}) f(0)
\end{equation*}
which gives 
\begin{equation}\label{k1}
  \int_{a\beta}^a x^{\alpha - 1} |f(x)|^p dx \geq |f(0)|^p \int_{a\beta}^a
  x^{\alpha - 1} \left|1 - \frac{x}{a\beta} \right|^p dx.
\end{equation}
Combining~\eqref{k1} and~\eqref{k2}, we obtain
\begin{equation*}
  \int_0^a x^{\alpha - 1} |f(x)|^p dx \geq |f(0)|^p  \int_0^{a} x^{\alpha
    - 1} \left|1 - \frac{x}{a\beta} \right|^p dx = |f(0)|^p a^{-p}
  \beta^{-p} \int_0^a x^{\alpha - 1} \left|x - a\beta \right|^p dx. 
\end{equation*}
Because $\beta < 1$, we get
\begin{equation*}
  \int_0^a x^{\alpha - 1} |f(x)|^p dx \geq |f(0)|^p a^{-p} \int_0^a 
  x^{\alpha - 1} |x - a\beta|^p dx.  
\end{equation*}
By the change of variable $x = a u$ and noting that $0 < \beta < 1$ is
arbitrary, we obtain 
\begin{equation*}
  \int_0^a x^{\alpha - 1} |f(x)|^p dx \geq |f(0)|^p a^{\alpha} \inf_{0
  \leq \beta \leq 1} \int_0^1 u^{\alpha - 1} |u - \beta|^p du.   
\end{equation*}
Because of~\eqref{sc}, this inequality also holds when $f(a) \leq
0$. This completes the proof.  
\end{proof}

We are now ready to prove Lemma~\ref{gbr}. 
\begin{proof}[Proof of Lemma~\ref{gbr}] 
  It is clear that, without loss of generality, we only need to
  prove~\eqref{gbr.eq} when $1/2 \leq y_i < 1$ for all $i = 1, \dots,
  n$. The hypothesis on $\phi$ implies that 
  \begin{equation*}
    \int_{y_1}^1 \dots \int_{y_d}^1 |\phi(x)|^p dx_1 \dots dx_n \leq
    1. 
  \end{equation*}
  We shall write the integral above in polar coordinates. Let 
  \begin{equation*}
    x_1 = y_1 + r \cos \theta_1, ~ x_2 = y_2 + r \sin \theta_1 \cos \theta_2,
    \dots, ~x_d = y_d + r \sin \theta_1 \dots \sin \theta_{d-2} \sin
    \theta_{d-1}.  
  \end{equation*}
Then 
  \begin{equation}\label{nam}
    \int_{\theta \in \Theta} \int_0^{r_{\theta}} |g(r, \theta)|^p
    r^{d-1} \sin^{d-2} \theta_1 \dots \sin \theta_{d-2} ~ dr
    d\theta_1 \dots d\theta_{d-1} =   \int_{y_1}^1 \dots \int_{y_d}^1
    |\phi(x)|^p dx \leq 1.  
  \end{equation}
for some set $\Theta$ with 
  \begin{equation*}
    g(r, \theta) := \phi(y_1 + r \cos \theta_1, ~ y_2
    + r \sin \theta_1 \cos \theta_2, \dots, ~y_n + r \sin
    \theta_1 \dots \sin \theta_{d-2} \sin \theta_{d-1}) 
  \end{equation*}
and 
\begin{equation*}
  r_{\theta} :=  \min \left(\frac{1-y_1}{\cos \theta},
    \frac{1-y_2}{\sin \theta_1 \cos \theta_2}, \dots,
    \frac{1-y_d}{\sin \theta_1 \dots \sin \theta_{d-2} \sin
      \theta_{d-1}} \right).  
\end{equation*}

Because of the convexity of $\phi$, the function $r \mapsto g(r,
\theta)$ is clearly convex on $(0, r_{\theta})$. Thus by
Lemma~\ref{fidi}, we obtain that for every $\theta \in \Theta$, 
\begin{equation*}
  \int_0^{r_{\theta}} |g(r, \theta)|^p r^{d-1} dr \geq C(d, p) |g(0,
  \theta)|^p r_{\theta}^d = d C(d, p) |\phi(y)|^p \int_0^{r_{\theta}}
  r^{d-1} dr  
\end{equation*}
where $C(d, p)$ is defined as in~\eqref{fidi.c}. We thus obtain
from~\eqref{nam} that 
\begin{equation*}
1 \geq d C(d, p) |\phi(y)|^p \int_{\theta \in \Theta}
\int_0^{r_{\theta}} r^{d-1}  \sin^{d-2} \theta_1 \dots \sin \theta_{d-2}~ dr
d\theta_1 \dots d\theta_{d-1}
\end{equation*}
Converting the above integral back to the regular coordinates, we get 
\begin{equation*}
  1 \geq d C(d, p) |\phi(y)|^p \int_{y_1}^1 \dots \int_{y_n}^1 dy_1
  \dots dy_d = d C(d, p) |\phi(y)|^p (1 - y_1) \dots (1 - y_d). 
\end{equation*}
This proves~\eqref{gbr.eq} with $c := d^{-1/p} C(d, p)^{-1/p}$. 
\end{proof}

\bf{Acknowledgement: } The author is sincerely thankful to the
anonymous referee whose comments led to an improvement of the paper.

\bibliographystyle{plainnat}
\bibliography{AG}

\end{document}